\documentclass[letterpaper,10 pt, conference]{ieeeconf}

\IEEEoverridecommandlockouts                              % This command is only needed if 
\usepackage{graphics} % for pdf, bitmapped graphics files
\usepackage{epsfig} % for postscript graphics files

\usepackage{amsmath,amssymb,amsfonts,amsthm} % assumes amsmath package installed
\usepackage{mathtools,bbm}
\usepackage{tikz, epic,eepic}
\usepackage{caption}
\usepackage{subcaption}
\usetikzlibrary{shapes,arrows}
\usetikzlibrary{automata}
\usepackage{bbm}
\usepackage{bm}
\usepackage{bbold}
\usepackage{epsfig, amsbsy}
\usepackage{mathtools}
\usepackage{graphicx}
\usepackage{caption}
\usepackage{subcaption}
\usepackage{algorithm}
\usepackage[noend]{algpseudocode}
\usepackage{epstopdf}
\usepackage{multirow,array}
\usepackage{cite}

\let\labelindent\relax
\usepackage{enumitem}
\theoremstyle{remark}

\newtheorem{theorem}{Theorem}
\newtheorem{proposition}{Proposition}
\newtheorem{corollary}{Corollary}
\newtheorem{lemma}{Lemma}

{\itshape}{\rmfamily}

\newcommand{\argmax}[1]{\underset{#1}
{\operatorname{arg}\,\operatorname{max}}\;}

\def\mcal{\mathcal}
\def\mbb{\mathbb}
\def\vbar{\bar{v}}
\def\vubar{\underline{v}}

\title{\LARGE \bf
Characterizing the interplay between information and strength in Blotto games
}

\author{Keith Paarporn, Rahul Chandan, Mahnoosh Alizadeh, and Jason R. Marden % <-this % stops a space
\thanks{K. Paarporn ({\tt\small kpaarporn@ucsb.edu}), R. Chandan ({\tt\small rchandan@ucsb.edu}), M. Alizadeh ({\tt\small alizadeh@ucsb.edu}), and J. R. Marden ({\tt\small jrmarden@ece.ucsb.edu}) are with the Department of Electrical and Computer Engineering at the University of California, Santa Barbara, CA.} \thanks{This work is supported by UCOP Grant LFR-18-548175, ONR grant \#N00014-17-1-2060, and NSF grant \#ECCS-1638214.}
}

\begin{document}

\maketitle
\thispagestyle{empty}
\pagestyle{empty}

%%%%%%%%%%%%%%%%%%%%%%%%%%%%%%%%%%%%%%%%%%%%%%%%%%%%%%%%%%%%%%%%%%%%%%%%%%%%%%%%
% In particular, an adversary well-informed about the details of the competition can be a more formidable opponent, even when it has fewer resources at its dispatch.
\begin{abstract}

In this paper, we investigate informational asymmetries in the Colonel Blotto game, a game-theoretic model of competitive resource allocation between two players over a set of battlefields. The battlefield valuations are subject to randomness. One of the two players knows the valuations with certainty. The other knows only a distribution on the battlefield realizations. However, the informed player has fewer resources to allocate. We characterize unique equilibrium payoffs in a two battlefield setup of the Colonel Blotto game. We then focus on a three battlefield setup in the General Lotto game, a popular variant of the Colonel Blotto game. We characterize the unique equilibrium payoffs and mixed equilibrium strategies. We quantify the value of information - the difference in equilibrium payoff between the asymmetric information game and complete information game. We find information strictly improves the informed player's performance guarantee. However, the magnitude of improvement varies with the informed player's strength as well as the game parameters. Our analysis   highlights the interplay between strength and information in adversarial environments.

\end{abstract}
%%%%%%%%%%%%%%%%%%%%%%%%%%%%%%%%%%%%%%%%%%%%%%%%%%%%%%%%%%%%%%%%%%%%%%%%%%%
%%%%%%%%%%%%%%%%%%%%%%%%%%%%%  S E C T I O N 1 %%%%%%%%%%%%%%%%%%%%%%%%%%%%%%%%%%%%%%
%%%%%%%%%%%%%%%%%%%%%%%%%%%%%%%%%%%%%%%%%%%%%%%%%%%%%%%%%%%%%%%%%%%%%%%%%%%
\section{Introduction}\label{sec:intro}
% !TEX root = arxiv.tex

% For instance, in the US presidential election, the number of electoral votes awarded by winning each state is known to both candidates.
In adversarial interactions, informational asymmetries may provide an advantage to one competitor over the other. Adversarial contests appear in a wide range of applications, such as political campaigns \cite{Behnezhad_2018,Thomas_2018}, security of cyber-physical systems \cite{Ferdowsi_2017}, and competitive advertising \cite{Fazeli_2017}. Rigorous studies on the interplay between asymmetric information and strategic decision-making in adversarial settings has received attention in the recent literature \cite{Li_2017,Kartik_2019}. Though these works analyze general zero-sum game settings,  another popular framework for analyzing such contests is the Colonel Blotto game. Two players strategically allocate their limited resources over a finite set of battlefields.  A player secures a battlefield if it  has allocated more resources to it than the opponent. Each player aims to secure as many battlefields as possible. The goal of the present paper is to quantify the performance improvements one player in the Blotto game may experience when it possesses system-level information while its opponent does not.

Though simply posed, the Colonel Blotto game features highly complex strategies.  The established literature on the Colonel Blotto game and its variants primarily focuses on characterizing mixed-strategy Nash equilibria, due to non-existence of pure equilibria in most cases of interest. First introduced by Borel in 1921 \cite{Borel}, researchers have incrementally contributed to this body of work over the last one hundred years \cite{Gross_1950}, \cite{Roberson_2006}. In 1950, Gross and Wagner \cite{Gross_1950} provided an equilibrium solution of the two battlefield case with asymmetric values and resources, as well as for $n$ homogeneous battlefields and symmetric resources. In 2006, the work of Roberson \cite{Roberson_2006} generalized the solution to $n$ battlefields and asymmetric forces by leveraging the solutions of all-pay auctions and the theory of copulas \cite{Sklar_1973}. 

%In 1950, Gross and Wagner  provided an equilibrium solution of the two battlefield case with asymmetric values and resources, as well as for $n$ homogeneous battlefields and symmetric resources. In 2006, the work of Roberson  generalized the solution to $n$ homogeneous battlefields and asymmetric forces by leveraging the solutions of all-pay auctions and the theory of copulas \cite{Sklar_1973}. 

There recently has been renewed interest in the Colonel Blotto game, along with its many variants and extensions \cite{Hart_2008,Macdonell_2015,Schwartz_2014,Kovenock_2015,Ferdowsi_2018,Thomas_2018}. Notably, results for heterogeneous battlefield valuations have  appeared in Colonel Blotto as well as the General Lotto variant, which admits a more tractable analysis \cite{Schwartz_2014,Kovenock_2015}.   Blotto games have also received attention in engineering application areas such as the security of cyber-physical systems \cite{Ferdowsi_2017} and the formation of large-scale networks \cite{Shahrivar_2014,Guan_2019}.  The vast majority of these studies assume the players have complete information about the opposing player's resource budget and the values of each battlefield.

% Roberson leverages copula theory \cite{Sklar_1973} to show a joint distribution that satisfies the budget constraints can be found using the marginal distributions obtained from the all-pay auctions.  
%The forces deployed to each battlefield can be modeled with integer  or continuous  allocations. 

Few works have considered the Blotto or Lotto game with incomplete information. In \cite{Adamo_2009}, the authors study a Blotto game in which players have incomplete information about the other's resource budgets. In \cite{Kovenock_2011}, the players are subject to incomplete information about the battlefield valuations. In both of these works, all players are equally uninformed about the parameters, and hence symmetric Bayes-Nash equilibria are provided. To the best of our knowledge, one-shot Blotto games where the players possess asymmetric information has not received attention in the literature. 
%A question we seek to address is whether a player with more information can defeat an opponent with more resources at its disposal. 

%Blotto games where the players possess asymmetric information is a subject that has yet to receive much attention. 

 In this paper, we formulate a Bayesian game framework in which one player is completely informed and the other does not receive any side information about the battlefield valuations (Section \ref{sec:model}).  Our main contributions characterize unique equilibrium payoffs and strategies in representative scenarios of this framework. We first analyze the  Colonel Blotto game with two battlefields (Section \ref{sec:results_CB}).  With three battlefields, we are able to solve for unique mixed-strategy equilibria in a representative General Lotto game under the same informational asymmetry (Section \ref{sec:results_GL}). We do this by leveraging established results on all-pay auctions with asymmetric information  \cite{Siegel_2014} (Section \ref{sec:APA}). We can then quantify the difference between the equilibrium payoff in this game and the scenario where both players are uninformed.

%First focusing on a two battlefield case, we analyze the Colonel Blotto game, calculating the equilibrium payoff by providing an asymmetric mixed-strategy equilibrium (Section \ref{sec:results_CB}). We then characterize the equilibrium payoff difference between this game and the scenario where both players are uninformed, quantifying a value of information.  In a three battlefield case, we analyze the General Lotto game (Section \ref{sec:results_GL}). The General Lotto game relaxes the budget constraint of the Blotto game, which requires resource allocations to be budget-feasible with probability one, to requiring allocations be feasible only in expectation. This relaxation increases the analytic tractability of the problem \cite{Thomas_2018,Kovenock_2015}. Here, we solve for unique mixed-strategy equilibria by leveraging an established result on asymmetric information all-pay auctions \cite{Siegel_2014}, and calculate the equilibrium payoff. In order to solve the General Lotto game, we leverage known results in two player all-pay auctions with asymmetric information and show under some conditions, these results may be applied to the General Lotto game (Section \ref{sec:APA}). 
 
These characterizations allow us to determine conditions under which the informed player has an advantage in the Lotto game, despite having fewer resources to allocate. As an illustrative analysis, we consider a scenario where both players are uninformed and the option of purchasing information with a fraction of its budget is available only to the weaker player. We quantify a measure of informational value determined by the largest proportion of  resources the player is willing to give up in exchange for information (Section \ref{sec:results_GL}).

% In Section \ref{sec:examples}, we provide simple examples to illustrate the complex interactions of a Colonel Blotto game, and why mixed strategies are needed in its analysis.

%
\begin{figure*}[t]
	\centering
	\begin{subfigure}{.32\textwidth}
		\centering
		\includegraphics[scale=.95]{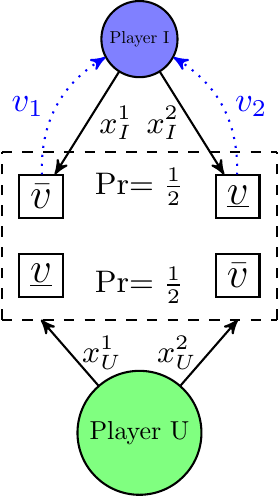}
		\caption{}
		\label{fig:blotto_diagram}
	\end{subfigure}  
	\begin{subfigure}{.32\textwidth}
		\centering
		\includegraphics[scale=.4]{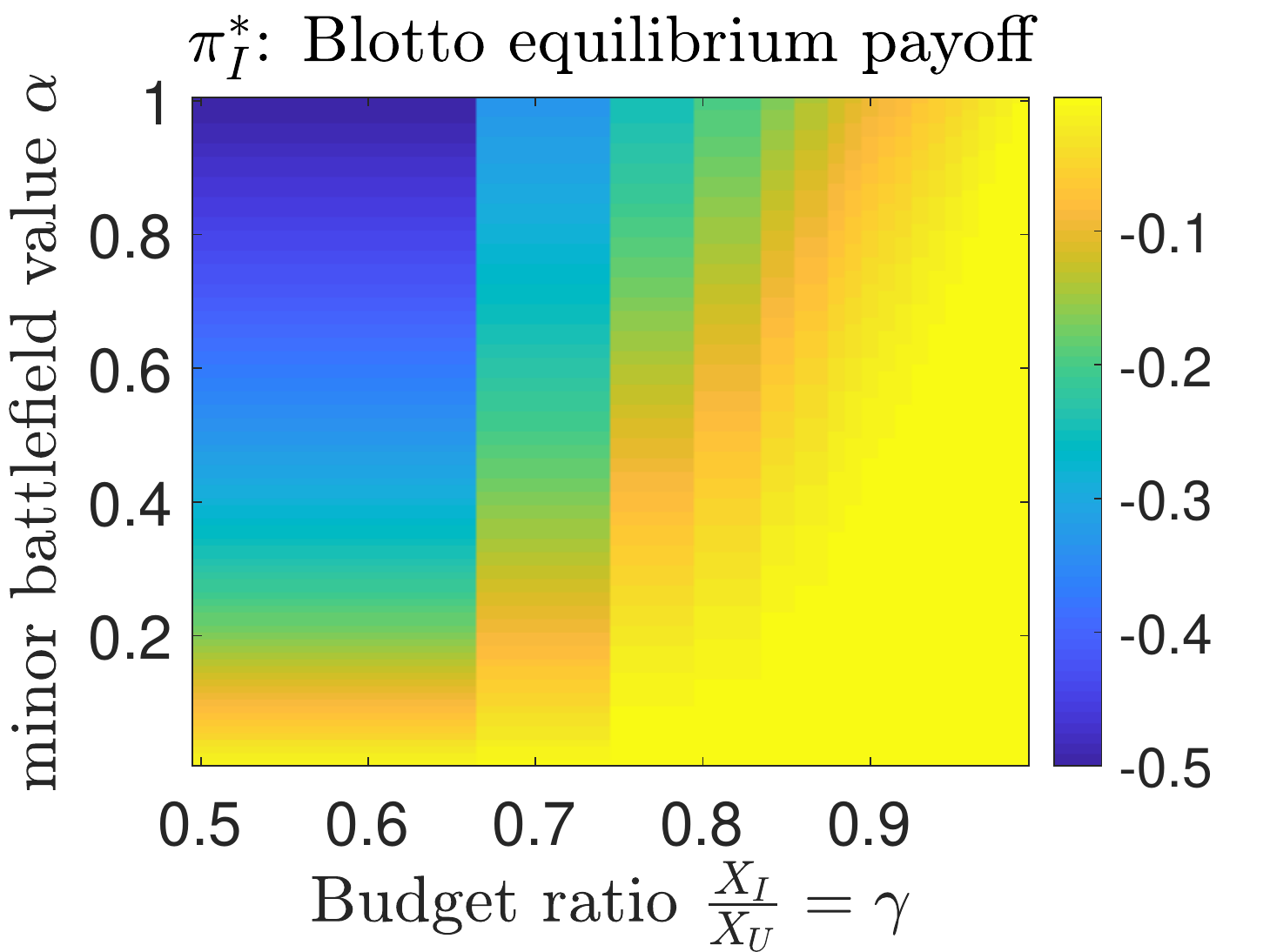}
		\caption{}
		\label{fig:blotto_pi_I}
	\end{subfigure}  
	\begin{subfigure}{.32\textwidth}
		\centering
		\includegraphics[scale=.4]{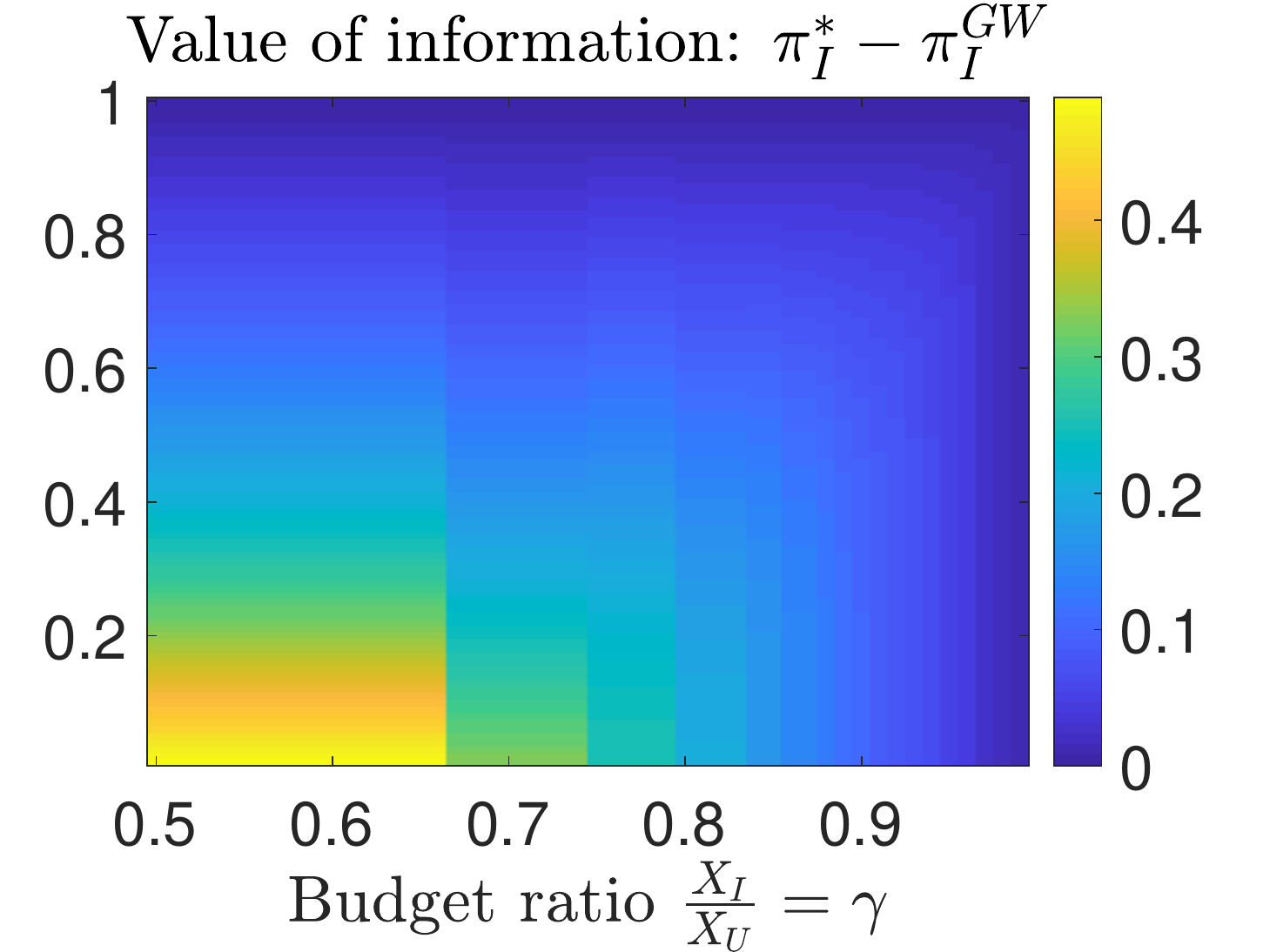}
		\caption{}
		\label{fig:blotto_VoI}
	\end{subfigure}
	\caption{\small (a) A diagram of the Colonel Blotto game with two battlefields and asymmetric budgets and information. There are two possible battlefield valuation sets, each realized with probability $\frac{1}{2}$. Player I has fewer resources than player U, but knows the true realization. (b) The equilibrium payoff \eqref{eq:blotto_payoff_midrange} to the informed player in the Blotto game $\text{CB}(X_I,X_U,V,\frac{1}{2})$, where $\vbar=1$ and $\vubar=\alpha\in(0,1)$. It is negative when $\gamma \in (\frac{1}{2},1)$ and $\alpha \in (0,1)$. (c) The value of information, $\pi_I^* - \pi_I^{\text{GW}}$, is non-negative. Information offers the most improvement for low budget ($\gamma$ near $\frac{1}{2}$) and when there is higher priority in the diagonal battlefields ($\alpha$ low).  } 
	\label{fig:blotto_plots}
\end{figure*}
%

%%%%%%%%%%%%%%%%%%%%%%%%%%%%%%%%%%%%%%%%%%%%%%%%%%%%%%%%%%%%%%%%%%%%%%%%%%%
%%%%%%%%%%%%%%%%%%%%%%%%%%%%%  S E C T I O N 2 %%%%%%%%%%%%%%%%%%%%%%%%%%%%%%%%%%%%%%
%%%%%%%%%%%%%%%%%%%%%%%%%%%%%%%%%%%%%%%%%%%%%%%%%%%%%%%%%%%%%%%%%%%%%%%%%%%
\section{Model}\label{sec:prelim}
%%%%%%%%%%%%%%%%%%%%%%%%%%%%%%%%%%%%%%%%%%%%%%%%%%%%%%%%%%%%%%%%%%%%

Two players, $I$ and $U$, allocate their non-atomic forces across $n$ battlefields simultaneously. In each battlefield $j\in\{1,\ldots,n\}$, the player that sends more forces wins the battlefield and receives payoff $v^j$, while the losing player receives $-v^j$. In the case of a tie on a battlefield, both players receive zero payoff. The utility to each player is the sum of payoffs across all battlefields.   The players have a budget on their forces, $X_I, X_U > 0$. We assume $X_I < X_U$. The budgets and battlefield valuations $\bm{v}:=\{v^j\}_{j=1}^n$ are common knowledge. For $z\in\{I,U\}$, a pure strategy is a non-negative vector $\bm{x}_z = (x^1_z,\ldots,x^n_z) \in \mbb{R}^n_+$. We call $x^z$ an \emph{allocation}. The payoff to player $z$  is
\begin{equation}\label{eq:u}
	u_z(\bm{x}_z,\bm{x}_{-z};\bm{v}) := \sum_{j=1}^n v^j \text{sgn}(x^j_z - x^j_{-z}) ,
\end{equation}
where we follow the convention $\text{sgn}(0) = 0$. This defines a zero-sum game since the payoff to player $-z$ is the negative of the above.  A mixed strategy for player $z$ is an $n$-variate distribution function $F_z$ on $\mbb{R}_+^n$. That is, an allocation $\bm{x}_z$ is drawn from the distribution $F_z$.  A distribution $F_z$ has $n$ univariate marginal  (cumulative) distributions $\{F_z^j\}_{j=1}^n$, specifying the allocation distributions to each battlefield. Extending the definition of \eqref{eq:u} to mixed strategies, the expected payoff for each player can be expressed as
\begin{equation}\label{eq:lotto_payoff}
	\begin{aligned}
		u_z(F_z,F_{-z};\bm{v}) &= \int_{\mbb{R}_+^n} \int_{\mbb{R}_+^n}\sum_{j=1}^n v^j \text{sgn}(x^j_z - x^j_{-z}) dF_z dF_{-z} \\
		&=\sum_{j=1}^n v^j  \int_0^\infty (2F_{-z}^j -1)dF_z^j .
	\end{aligned}
\end{equation}
\vspace{-5mm}
\subsection{The Colonel Blotto game}
In the Colonel Blotto game, a feasible allocation $\bm{x}_z$ for player $z \in \{I,U\}$  lies in the set
\begin{equation}\label{eq:blotto_constraint}
	B(X_z) := \Bigg\{ \bm{x}_z \in \mbb{R}_+^n : \sum_{j=1}^n x^j_z = X_z \Bigg\} . \tag{BC}
\end{equation}
Furthermore, the support of a mixed strategy $F_z$ is contained in $B(X_z)$. Thus, an allocation $\bm{x}_z$ drawn from $F_z$ belongs to $B(X_z)$ with probability one. We denote the set of all such feasible mixed strategies with $\mcal{B}(X_z)$. We refer to the Colonel Blotto game with $\text{CB}(X_I,X_U,\bm{v})$. 
%The following example illustrates that mixed strategies are instrumental in the analysis of Colonel Blotto games.
%\begin{example}[Non-existence of pure strategy Nash equilibria]
%	There are two battlefields both worth a value of one. Player $I$ has a budget $\frac{X_U}{2}<X_I \leq X_U$. A pure strategy for player $I$ is an allocation $x_I \in [0,X_I]$ to send to battlefield 1, and the rest $X_I - x_I$ to battlefield 2. Likewise, player $U$ chooses $x_U \in [0,X_U]$.  We demonstrate that there is no pure Nash equilibrium in this setting. For any pure allocation $x_I$ of $I$, player U can secure both battlefields by responding with $x_U = x_I + \frac{X_U-X_I}{2}$, ensuring a payoff of 2. For any pure allocation $x_U$, $U$ puts less than or equal to $\frac{X_U}{2}$ on one of the two battlefields. In response, player $I$ can secure that battlefield by sending all forces to it, since $X_I >\frac{X_U}{2}$. Both players win and lose one battlefield, and the payoff is zero. Player $U$ can now switch to a pure strategy to win both battlefields, which breaks the conditions for equilibrium. Hence, no two allocations $x_I, x_U$ are best-responses to each other.
%\end{example}

\subsection{The General Lotto game}

The General Lotto game relaxes the support constraint in the Colonel Blotto game, requiring each player's allocation budget to be met only in expectation with respect to their marginal distributions. That is,
\begin{equation}\label{eq:lotto_constraint}
	\mcal{L}(X_z) := \Bigg\{ F_z : \sum_{j=1}^n \mbb{E}_{F_z^j}[x^j_z] = X_z \Bigg\} \tag{LC}
\end{equation}
is the set of such feasible mixed strategies. Note that both the budget constraint \eqref{eq:lotto_constraint} and expected payoff \eqref{eq:lotto_payoff} only depend on the univariate marginal distributions and not the joint $n$-variate distribution $F_z$. Hence, a player's choice amounts to selecting $n$ independent univariate marginals that satisfy \eqref{eq:lotto_constraint}. We specify General Lotto game with $\text{GL}(X_I,X_U,\bm{v})$.

\section{Blotto and Lotto games with asymmetric information}\label{sec:model}
% !TEX root = arxiv.tex

Using a standard Bayesian games framework, we wish to model a situation in which the battlefield valuations are subject to randomness. Suppose there are $m$ possible sets of battlefield valuations, labeled by the states  $\Omega = \{\omega_1,\ldots,\omega_m\}$.  The state $\omega_i$ is realized with probability $p_i > 0$, and the probability vector $\bm{p}\in\Delta(\Omega)$ satisfying $\sum_{i=1}^m p_i = 1$ is common knowledge to the players.  Given state $\omega_i$ is realized, battlefield $j$ is worth $v_i^j > 0$, and we denote $V$ as the $m \times n$ matrix with elements $v_i^j$.  

We consider the setting where player $I$ observes the true state realization, and player $U$ does not receive any side information about the realization\footnote{A more general framework of asymmetric information may be formulated with Bayesian games. That is, arbitrary partial information structures may be assigned to the players. Since our focus is on one informed and one uninformed player, we leave such generalizations to future work.}. We call $U$ the ``uninformed" player. Specifically, player $I$ has $m$ distinct types $T_I := \{t_1,\ldots,t_m\}$, receiving type $t_i$ whenever state $\omega_i$ is realized. That is, $I$ knows the realization with certainty. We refer to $T_I$ as its type space. Player $U$ has a single type, and hence infers the realization (and $I$'s type) according to the common prior $\bm{p}$. The information structure is known to both players. 

A strategy for player $I$ consists of $n$-variate distributions $\{F_I(t_i)\}_{i=1}^m$, one for each type $t_i$, $i=1,\ldots,m$. We refer to the collection $F_I = \{F_I(t_i)\}_{i=1}^m$ as $I$'s strategy. Player $U$'s strategy $F_U$ is a single $n$-variate distribution. For the Lotto game, $F_I(t_i) \in \mcal{L}(X_I)$ for each $t_1,\ldots,t_m$, and $F_U \in \mcal{L}(X_U)$. For the Blotto game, they belong to $\mcal{B}(X_I)$ and $\mcal{B}(X_U)$, respectively. We denote $\{F_I^j(t_I)\}_{j=1}^n$ as the univariate marginals of $F_I(t_I)$ for each $t_I \in T_I$, and $F_U^j$ the univariate marginals of $F_U$. All of these components specify a Lotto or Blotto game with incomplete information, which we denote as $\text{GL}(X_I,X_U,V,\bm{p})$ and $\text{CB}(X_I,X_U,V,\bm{p})$, respectively.  The expected utility of player $U$ is
\begin{equation}\label{eq:payoff_U}
	%\pi_U(F_U,F_I | t_U) = \sum_{j=1}^n \left[ \sum_{i=1}^m p_i v_i^j \int_0^\infty   \left( 2F_I^j(t_i) - 1 \right) dF_U^j \right]
	\pi_U(F_U,F_I) := \sum_{i=1}^m p_i u_U(F_U,F_I(t_i),\omega_i)
\end{equation}
where with some abuse of notation, we use $\omega_i$ to refer to the battlefield valuation set in state $i$. The ex-interim utility for player $I$, given type $t_i$ is the expected payoff
\begin{equation}\label{eq:interim_I}
	%\pi_I(F_I,F_U | t_i) := \sum_{j=1}^n  v_i^j \int_0^\infty (2F_U^j-1)  dF_I^j(t_i) .
	\pi_I(F_I(t_i),F_U | t_i) := u_I(F_I(t_i),F_U,\omega_i) .
\end{equation}
A Bayes-Nash equilibrium (BNE) is a strategy profile $(F_I^*,F_U^*)$ satisfying 
\begin{equation}
	\begin{aligned}
		F_I^*(t_i) &\in \argmax{F_I(t_i)} \pi_I(F_I(t_i),F_U^* | t_i), \forall i=1,\ldots,m \\
		\text{and } F_U^* &\in \argmax{F_U} \pi_U(F_U,F_I^*).
	\end{aligned}
\end{equation}
An equivalent BNE condition is that $(F_I^*,F_U^*)$ satisfies the best-response correspondences in ex-ante payoffs, given that each $p_i > 0$ \cite{Vega-Redondo_2003}. That is, $F_U^* \in \argmax{F_U} \pi_U(F_U,F_I^*)$ and $F_I^* \in \argmax{\{F_I(t_i)\}_{i=1}^m} \pi_I(F_I,F_U^*)$, where $I$'s ex-ante payoff is
\begin{equation}\label{eq:I_exante}
	\pi_I(F_I,F_U) := \sum_{i=1}^m p_i \pi_I(F_I(t_i),F_U | t_i).
\end{equation}  
%\begin{equation}
%	\text{BR}_I(F_U) := \argmax{\{F_I(t_i)\}_{i=1}^m} \pi_I(F_I,F_U) .
%\end{equation}
We note that since the underlying  complete information game is zero-sum, the games with asymmetric information are also zero-sum with respect to ex-ante utilities \eqref{eq:payoff_U}, \eqref{eq:I_exante}. This allows us to speak of a unique equilibrium ex-ante payoff $\pi_I^* = -\pi_U^*$ of the Blotto and Lotto games with asymmetric information.

%\subsection{Completely informed versus uninformed}
%
%We are interested in a special case of the Blotto and Lotto games with incomplete information where player $I$ is completely informed of the battlefield valuations and $U$ is completely uninformed. That is, player $I$ has $m$ types $T_I = \{t_i\}_{i=1}^m$ and its signal function is $\tau_I(\omega_i) = t_i$ for every $i=1,\ldots,m$. Player $U$ has a single type, $T_U = \{t_U\}$. Hence, the belief for $I$ is $\text{Pr}(\omega_k,t_U | t_i) = \mbb{1}(i=k)$, and the belief for $U$ is $\text{Pr}(\omega_k,t_i | t_U) = p_k \mbb{1}(i=k)$. We denote this particular information structure with the signal functions $\hat{\tau}$.

%when $v_{11} = v_{22} = \frac{1}{1+\alpha}$, $v_{12} = v_{21} = \frac{\alpha}{1+\alpha} $, $\alpha \in (0,1)$, and $p=\frac{1}{2}$

%%%%%%%%%%%%%%%%%%%%%%%%%%%%%%%%%%%%%%%%%%%%%%%%%%%%%%%%%%%%%%%%%%%%%%%%%%%
%%%%%%%%%%%%%%%%%%%%%%%%%%%%%  S E C T I O N 3 %%%%%%%%%%%%%%%%%%%%%%%%%%%%%%%%%%%%%%
%%%%%%%%%%%%%%%%%%%%%%%%%%%%%%%%%%%%%%%%%%%%%%%%%%%%%%%%%%%%%%%%%%%%%%%%%%%
\section{Colonel Blotto results}\label{sec:results_CB}
In the following structural result, we characterize sufficient conditions on the budget ratio $\gamma := \frac{X_I}{X_U}$ such that the stronger player $U$ is guaranteed a positive equilibrium payoff regardless of the information asymmetry.
%we characterize upper bounds on the budget asymmetry $\gamma := \frac{X_I}{X_U}$ depending on the number of battlefields $n$ such that the stronger player $U$ derives a positive equilibrium payoff regardless of the information asymmetry.
%characterizing the minimum budget necessary for which player $I$ can derive a positive equilibrium payoff when there are $n$ battlefields.
\begin{proposition}
	Assume $\frac{X_I}{X_U} < 1$. Consider the game $\text{CB}(X_I,X_U,V,\bm{p})$ where $V \in m\times n$, i.e. there are $n$ battlefields and $m$ possible state realizations. A sufficient condition for which $\pi_U^* \geq 0$ is
	\begin{equation}
		\begin{cases}
			\frac{X_I}{X_U} < \frac{2}{n}, \quad &\text{if } n \text{ is even} \\
			\frac{X_I}{X_U} < \frac{2}{n+1}, \quad &\text{if } n \text{ is odd} \\
		\end{cases}
	\end{equation}
\end{proposition}
\begin{proof}
	Denote $E_j = \sum_{i=1}^m p_i v_i^j$ as the prior expected value of battlefield $j$. Without loss of generality, assume the $n$ battlefields are ordered according to $E_1 \geq E_2 \geq \cdots \geq E_n$. For $n$ even, consider player U's deterministic strategy that places $2X_I/n$ to the first $n/2$ battlefields, and the remaining resources allocated arbitrarily. This guarantees U attains at least half of the total available value. That is, U's security value for this strategy is at least zero. Similar reasoning applies to the case of $n$ odd.
\end{proof}
For two battlefields, the weaker resource player $I$ never wins the game, despite having better information. The following analysis focuses on the two battlefield case of $\text{CB}(X_I,X_U,V,\bm{p})$ under the uniform prior $\bm{p}=\frac{1}{2}\mbb{1}_2$. Although $I$ never wins, we quantify how information improves its payoff relative to when both players are uninformed. A solution for a general prior $\bm{p}=[p,1-p]$ remains a challenge in this analysis, due to the complexity of finding suitable equilibrium distributions. Hence for the rest of this section, we assume $p=\frac{1}{2}$. 
\subsection{Two battlefields case}
If $\frac{X_I}{X_U} < \frac{1}{2}$, player $U$ can secure both battlefields regardless of what $I$ does. We thus restrict our attention to the regime $\frac{X_I}{X_U} \in (\frac{1}{2},1)$.
We state the main result of this section, which characterizes the equilibrium payoff to $I$. Here, we assume a symmetric structure on the battlefields $V = \frac{1}{\vbar+\vubar}\begin{bmatrix} \vbar & \vubar \\ \vubar & \vbar \end{bmatrix}$ for $\vbar > \vubar > 0$. Figure \ref{fig:blotto_diagram} illustrates the setup of this game. %We also make use of the following notation, which is adapted from \cite{Gross_1950}. Write $d = X_U - X_I$. Let $q$ and $r$ be such that $X_U = qd + r$ where $0\leq r < d$ and $q$ is an integer. 
\begin{theorem}\label{blotto_phalf}
	Let $\frac{X_I}{X_U} \in (\frac{1}{2},1)$, and assume $p=\frac{1}{2}$. Define $q = \lfloor \frac{X_U}{X_U - X_I} \rfloor$. Then the ex-ante equilibrium payoff to the informed player $I$ is
	\begin{equation}\label{eq:blotto_payoff_midrange}
		\pi_I^* = \begin{cases}
			-\left(2\sum_{k=0}^{\frac{q-1}{2}} \left(\vbar/\vubar \right)^k -1 \right)^{-1} & \text{if } q \text{ odd} \\
			-\frac{\vubar}{\vbar+\vubar}\left(\sum_{k=0}^{\frac{q}{2}-1} \left(\vbar/\vubar \right)^k\right)^{-1} & \text{if } q \text{ even}
		\end{cases}
	\end{equation}
\end{theorem}

We provide the proof in the Appendix, which details a set of equilibrium mixed strategies. 
When both players are uninformed, i.e. both have a single type, we have a complete information game where the two battlefield valuations are their expected values, $\frac{1}{2}$. The equilibrium payoff is then given by an application of Gross and Wagner's solution  \cite{Gross_1950}, which yields $\pi_I^{\text{GW}} := -\frac{1}{q}.$ The following result verifies the equilibrium payoff improves when $I$ obtains information.
\begin{corollary}
	We have $\pi_I^* > \pi_I^{\text{GW}}$. That is, information strictly improves the equilibrium payoff for $I$.
\end{corollary}
\begin{proof}
	Consider the $q$ even case. Since $\frac{\vubar}{\vbar+\vubar} < \frac{1}{2}$ and $\frac{\vbar}{\vubar} > 1$ for $\vbar > \vubar$, from \eqref{eq:blotto_payoff_midrange} we get $-\frac{\vubar}{\vbar+\vubar}\left(\sum_{k=0}^{\frac{q}{2}-1} \left(\vbar/\vubar \right)^k\right)^{-1} > -\frac{1}{q}$. Similar arguments apply in the $q$ odd case.
\end{proof}

We may quantify a value of information in this setting as the payoff difference $\pi_I^* - \pi_I^{\text{GW}} $. Setting $\vbar=1$ and $\vubar = \alpha \in (0,1)$, we plot the equilibrium values $\pi_I^*$ in Figure \ref{fig:blotto_plots} , as well as the value of information. In the two battlefield Blotto game, an informed player still cannot gain an advantage, as $\pi_I^* < 0$ for all parameters $(\alpha,\gamma) \in (0,1)\times(0,1)$. Information offers the most payoff improvement  when $I$'s budget is low and the minor battlefield $\alpha$ is worth less (Fig. \ref{fig:blotto_VoI}).

In the next section, we consider the three battlefield case in the General Lotto game. In general, equilibrium solutions for $n \geq 3$ heterogeneous battlefields have not been characterized in the Colonel Blotto game, due to the complexity of finding suitable copulas for the marginal distributions \cite{Thomas_2018}. The Lotto constraint relaxation allows for analytic tractability in cases of more than two heterogeneous battlefields \cite{Kovenock_2015}.

%%%%%%%%%%%%%%%%%%%%%%%%%%%%%%%%%%%%%%%%%%%%%%%%%%%%%%%%%%%%%%%%%%%%%%%%%%%
\section{Results on General Lotto}\label{sec:results_GL}
%%%%%%%%%%%%%%%%%%%%%%%%%%%%%%%%%%%%%%%%%%%%%%%%%%%%%%%%%%%%%%%%%%%%%%%%%%%

%
\begin{figure*}[t]
	\centering
	\hspace{-3mm}
	\begin{subfigure}{.32\textwidth}
		\centering
		\includegraphics[scale=.31]{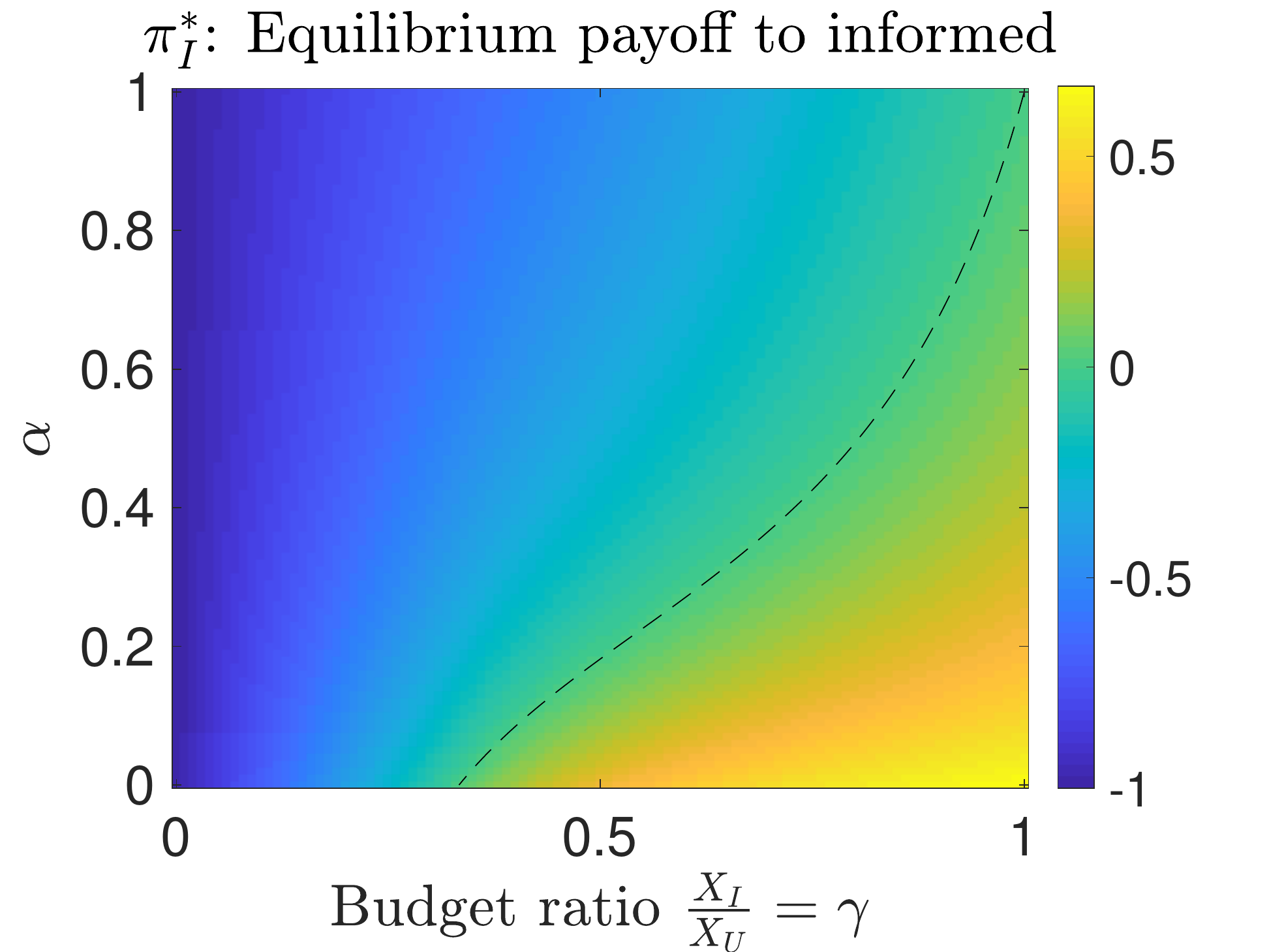}
		\caption{}
		\label{fig:lotto_pi_I}
	\end{subfigure}  
	\begin{subfigure}{.32\textwidth}
		\centering
		\includegraphics[scale=.31]{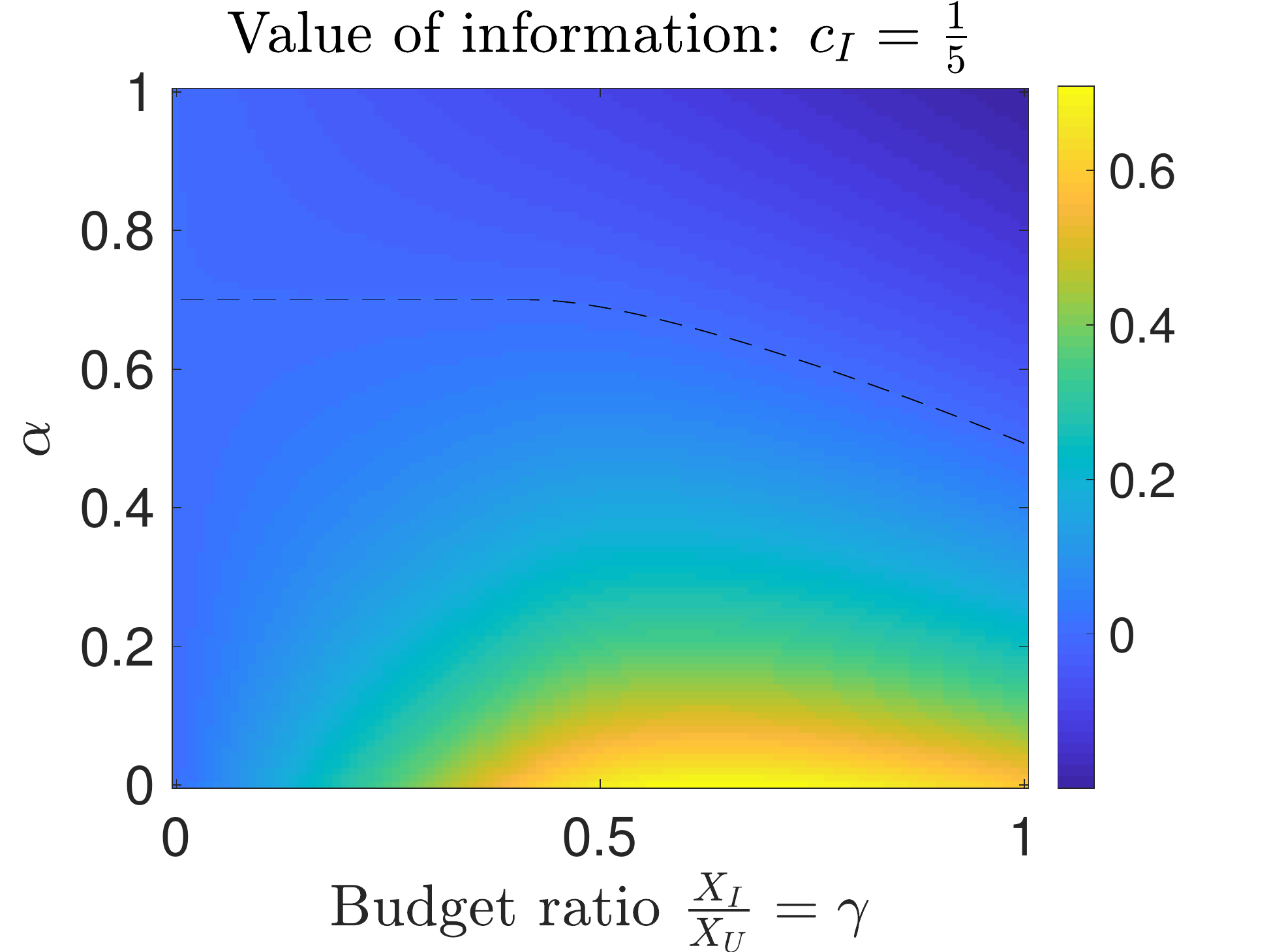}
		\caption{}
		\label{fig:lotto_VoI}
	\end{subfigure}
	\begin{subfigure}{.32\textwidth}
		\centering
		\includegraphics[scale=.31]{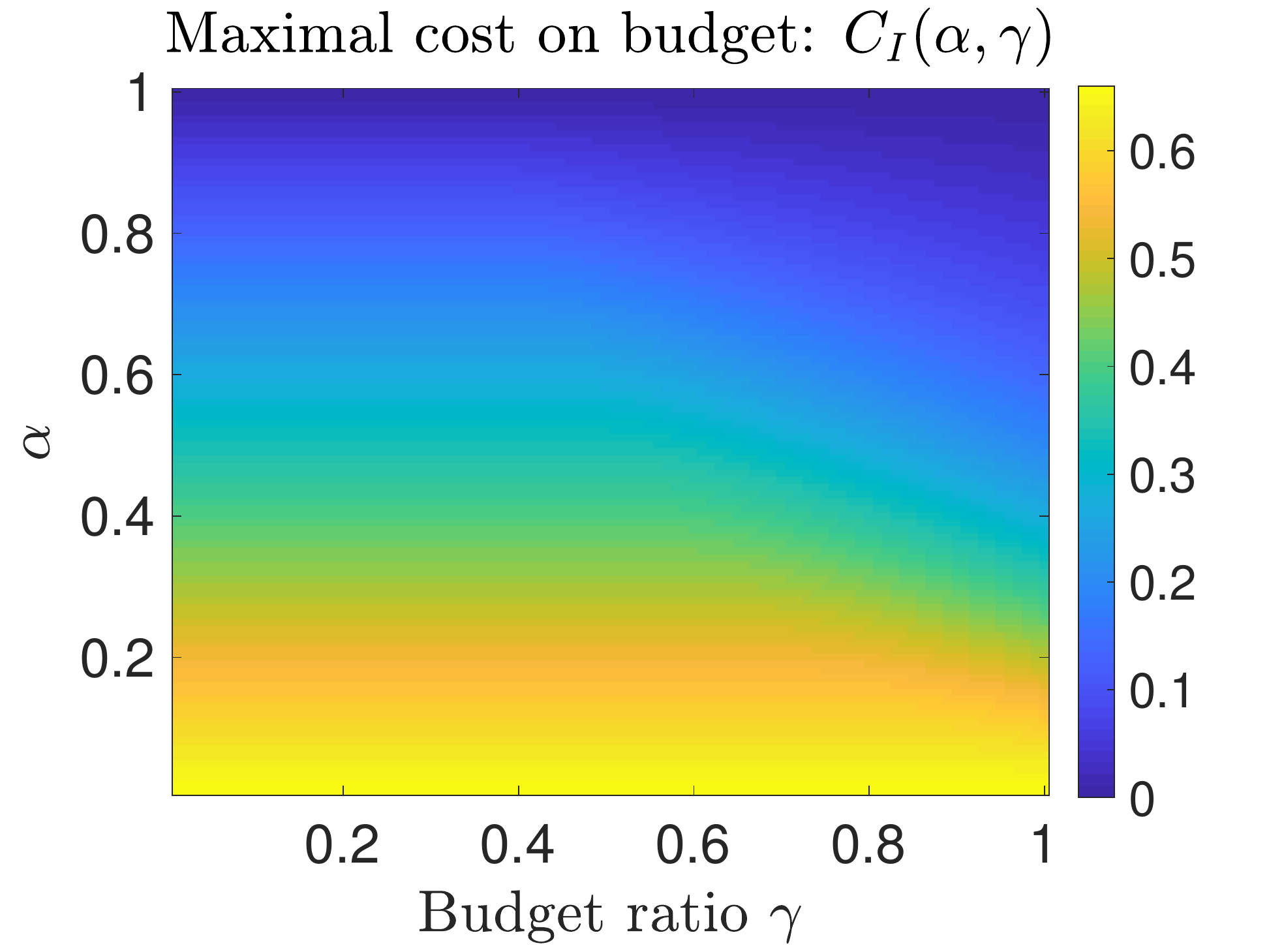}
		\caption{}
		\label{fig:lotto_bVoI}
	\end{subfigure}
	\caption{\small (a) The equilibrium payoff $\pi_I^*$  to the informed player in the Lotto game $\text{GL}(X_I,X_U,V_{\alpha\beta},\frac{1}{3}\mbb{1}_3)$, under the special case $\alpha = \beta$. The dashed line are the points $(\gamma,\alpha)$ for which $\pi_I^*=0$. That is, for parameters below the line, the informed player ``wins" the game and ``loses" for parameters above the line (see Corollary \ref{lemma_zerocrossing}). (b) The value of information, or payoff gain or loss to player $I$ from purchasing information at the cost $c_I = \frac{1}{5}$. The dashed black line in Figure \ref{fig:lotto_VoI} indicates the parameters for which $\text{VoI}(\alpha,\gamma) = 0$. (c) The maximal cost \eqref{eq:bVoI}, or the largest fraction of its resource budget $\gamma$ that player $I$ is willing to give up in exchange for information before it experiences a payoff loss. Information is more valuable in lower $\alpha$ ranges - when the battlefield rewards are concentrated at diverse locations.} 
	\label{fig:lotto_plots}
\end{figure*}

We restrict our attention to a representative three-battlefield Lotto game $\text{GL}(X_I,X_U,V_{\alpha\beta},\bm{p})$, where $\bm{p}=(\frac{1}{3},\frac{1}{3},\frac{1}{3})$, and $V_{\alpha\beta} = \frac{1}{1+\alpha+\beta}\begin{bmatrix} 1 & \alpha & \beta \\ \beta & 1 & \alpha \\ \alpha & \beta & 1 \end{bmatrix}$.
%\begin{equation}
%	V_{\alpha\beta} = \frac{1}{1+\alpha+\beta}\begin{bmatrix} 1 & \alpha & \beta \\ \beta & 1 & \alpha \\ \alpha & \beta & 1 \end{bmatrix}
%\end{equation}
%In this setting, the uninformed player does not know the location of the priority battlefield worth 1, nor the locations of the minor battlefields worth $\alpha$ and $\beta$. 
where $1 > \alpha \geq \beta > 0$. In each set of battlefields (rows), the total valuation is normalized to one. Though this formulation may be quite specific, the representative game serves as an illustrative and tractable scenario highlighting the informational asymmetry between players. Intuitively, for small $\alpha,\beta$ values, the valuable battlefield is worth 1 but sits at a  different location in each realization. An informed player would be able to take the most advantage by focusing its resources on this battlefield without wasting resources on the other battlefields. 

\subsection{Main result: characterization of equilibrium payoff}

\begin{theorem}\label{thm_lotto}
	Let $\gamma = \frac{X_I}{X_U} \leq 1$. Then player $I$'s equilibrium payoff, $\pi_I^*(\alpha,\beta,\gamma)$, for $\alpha,\beta,\gamma \in (0,1)$ takes the values
	\begin{equation}\label{eq:lotto_payoff1}
		\begin{array}{ll}
			\frac{3\gamma}{1+\alpha+\beta} - 1, \quad &\text{if } \gamma \in (0,\frac{1}{3}] \\ 				\frac{1}{1+\alpha+\beta} \left[ \left(1 \! -\! \frac{1}{3\gamma}\right)( 3\gamma\alpha + (1\!-\!\alpha)) + 1 \right] \!-\! 1, &\text{if } \gamma \in (\frac{1}{3},\frac{2}{3}] 
		\end{array}
	\end{equation}
	and
	\begin{equation}\label{eq:lotto_payoff2}
		\begin{array}{ll}
		\frac{1}{1\!+\!\alpha\!+\!\beta}\!\left[ 2-\frac{1}{3\gamma} + \alpha\left(2-\frac{1}{\gamma}\right) + 3\beta\gamma\left(1-\frac{2}{3\gamma}\right)\right] \!-\! 1
		\end{array}
	\end{equation}
	if $\gamma \in (\frac{2}{3},1]$.
\end{theorem}
A plot of $\pi_I^*$ is shown in Figure \ref{fig:lotto_pi_I} for the special case $\alpha=\beta$. We provide the details of the proof in Section \ref{sec:APA}, where we draw upon known results in asymmetric information all-pay auctions. An iterative algorithm is formulated in \cite{Siegel_2014} to construct mixed equilibrium strategies.  To verify that the set of constructed strategies indeed is a BNE of $\text{GL}(X_I,X_U,V_{\alpha\beta},\frac{1}{3}\mbb{1}_3)$, the constraint \eqref{eq:lotto_constraint} must be met.

\subsection{The value of information in the Lotto game}
In the following analysis, we consider a scenario in which both players are uninformed, and the option to purchase information with a fraction of its budget is available to player $I$. We quantify a value of information as the equilibrium payoff gain or loss in purchasing information, and specify the maximal cost $I$ is willing to pay before it experiences a payoff loss. In the case both players are uninformed, we may use the algorithm of \cite{Siegel_2014} to solve the Lotto game. In this setting, we arrive at the equilibrium payoff $\gamma-1$ for $I$. We first highlight some immediate consequences of Theorem \ref{thm_lotto}.  
%the largest proportion of its resources $I$ is  willing to give up in exchange for information in order to obtain an equilibrium payoff gain.
\begin{corollary}\label{lotto_cor1}
	We have that $\pi_I^*(\alpha,\beta,\gamma) > \gamma-1$ for $(\alpha,\beta,\gamma) \in (0,1)^3$. That is, information  strictly improves the equilibrium payoff for $I$. Additionally, $\pi_I^*$ is strictly decreasing in $\alpha$ and $\beta$, and strictly increasing in $\gamma$.
\end{corollary}
\begin{proof}
	The first claim follows from comparing the expressions \eqref{eq:lotto_payoff1} and \eqref{eq:lotto_payoff2} to $\gamma - 1$. The second claim follows by showing the respective signs of partial derivatives hold.
\end{proof}
We also characterize the parameter region in which an informed $I$ wins the game for the special case $\alpha=\beta$. 
%We may then quantify a value of information as the payoff difference $\pi_I^* - (\gamma-1)$, which we plot in \ref{fig:lotto_VoI} for the special case $\alpha=\beta$. From Corollary \ref{lotto_cor1}, the value of information decreases in $\alpha$. 
\begin{corollary}\label{lemma_zerocrossing}
	Fix a budget ratio $\gamma$. Then $\pi_I^*(\alpha,\alpha,\gamma) > 0$ if and only if $\alpha<\frac{\frac{1}{3}-\gamma}{3\gamma^2 - 4\gamma + \frac{1}{3}}$.
\end{corollary} 
\begin{proof}
	Player $I$ cannot win in the first region $\gamma \in (0,\frac{1}{3}]$, as $\pi_I^* \leq 0$. We find that the second expression of \eqref{eq:lotto_payoff1} is equivalent to \eqref{eq:lotto_payoff2} when $\beta=\alpha$. We then solve the equation $\pi_I^*(\alpha,\alpha,\gamma) = 0$ for $\alpha$, from which we obtain the result.
\end{proof}

%This suggests that information most improves the equilibrium payoff to $I$ when the battlefield values are more concentrated at diverse locations (Fig. \ref{fig:lotto_VoI}). 

 We now provide a general quantity of describing the value of information. In particular, we consider a scenario where both players are uninformed. Player $I$, with the budget ratio $\gamma$, has an opportunity to purchase information with a fraction $c_I$ of its budget. The value of information quantifies the equilibrium payoff gain or loss in purchasing the information at the budget fraction cost $c_I$. Formally, the value of information is the quantity
\begin{equation}
	\text{VoI}(\alpha,\gamma) := \pi_I^*(\alpha,\alpha,(1-c_I)\gamma) - (1-\gamma).
\end{equation}
An instance of the value of information is plotted in Figure \ref{fig:lotto_VoI}, when the cost is $c_I = \frac{1}{5}$. We note there is a regime in which information is not worth the cost, i.e. $\text{VoI}(\alpha,\gamma) < 0$. 
 
We also seek to find the highest cost on information that $I$ is willing to pay before it experiences an equilibrium payoff loss. To quantify this cost, let $\gamma_e$ be defined as the value that satisfies $\pi_I^*(\alpha,\alpha,\gamma_e) = \gamma-1$. Such a value is unique and well-defined for any $\alpha$, since $\pi_I^*(\alpha,\alpha,\gamma) > \gamma-1$ and is strictly increasing in $\gamma$, by Corollary \ref{lotto_cor1}. Then $C_I(\alpha,\gamma) := \frac{\gamma-\gamma_e}{\gamma}$ is the largest fraction of resources $I$ can give up. That is, for all $c_I \in [0,C_I(\alpha,\gamma)]$, we have $\pi_I^*(\alpha,\alpha,(1-c_I)\gamma) \geq \gamma-1$, with equality if and only if $c_I = C_I(\alpha,\gamma)$.  Then,
\begin{corollary}
	Fix a budget ratio $\gamma$. Then
    	\begin{equation}\label{eq:bVoI}
    		C_I(\alpha,\gamma) = 
		\begin{cases}
			\frac{1}{\gamma} (\gamma - \gamma_e(\alpha,\gamma)) \quad &\text{if } (\alpha,\gamma)\in B\\
			1-\frac{1}{3c_\alpha} &\text{if } (\alpha,\gamma)\notin B
		\end{cases}
    	\end{equation}
	where $\gamma_e(\alpha,\gamma) \!:=\! \frac{-(2(1-\alpha)c_\alpha - \gamma)+\sqrt{(2(1-\alpha)c_\alpha - \gamma)^2 + 4c_\alpha^2 \alpha(1-\alpha)}}{6\alpha c_\alpha}$, $c_\alpha := \frac{1}{1+2\alpha}$, and $B := \{(\alpha,\gamma) \in (0,1)^2 :  \gamma_e(\alpha,\gamma)  \geq \frac{1}{3} \}$.
\end{corollary}
\begin{proof}
	We solve the equation $\pi_I^*(\alpha,\alpha,\gamma_e) = \gamma-1$ for $\gamma_e$. Recall the second expression of \eqref{eq:lotto_payoff1} is equivalent to \eqref{eq:lotto_payoff2} when $\beta=\alpha$. The resulting value is valid only if $\gamma_e > \frac{1}{3}$. If not, we solve for $\gamma_e$ using the first entry of \eqref{eq:lotto_payoff1}.
\end{proof}
Figure \ref{fig:lotto_bVoI} plots $C_I(\alpha,\gamma)$. As an example, when $\gamma=1$ and $\alpha=0$, player $I$ can exchange up to $\frac{2}{3}$ of its budget for information and still obtain a payoff greater than $\gamma-1 = 0$. When $\alpha$ is high, player $I$ cannot afford to give up resources, as information is less valuable in this regime. 

%Note that the plot of Figure \ref{fig:lotto_VoI} portrays a special case of the quantity $\pi_I^*(\alpha,\alpha,(1-c_I)\gamma) - (\gamma-1)$ with $c_I = 0$. That is, it is the payoff gain when $I$ gets information for free. When $c_I > 0$, there will be regions in the space $(\gamma,\alpha)$ for which $\pi_I^*(\alpha,\alpha,(1-c_I)\gamma) - (\gamma-1)$ is negative - i.e. there are some cases where purchasing information is not worth it.

%For intermediate ranges of budget exchange, the payoff gain can be characterized by the payoff gain $\pi_I^* - (\gamma-1)$

%%%%%%%%%%%%%%%%%%%%%%%%%%%%%%%%%%%%%%%%%%%%
\section{Proof of Theorem \ref{thm_lotto}}\label{sec:APA}
% !TEX root = arxiv.tex
% The proof of the Theorem \ref{thm_lotto} requires using an established result in two-player all-pay auctions with asymmetric information and valuations.
We show that the equilibria in the Lotto game $\text{GL}(X_I,X_U,V,\bm{p})$ coincides with equilibria of $n$ independent all-pay auctions with asymmetric information and valuations. The work of Siegel \cite{Siegel_2014} provides an iterative algorithm for which to construct equilibrium mixed strategies in the all-pay auction. We leverage this procedure to construct proposed equilibrium distributions in the Lotto game, and verify that they satisfy the constraint \eqref{eq:lotto_constraint}. 

\subsection{Two-player all-pay auctions with asymmetric information and valuations}
Here, we specialize the setup of \cite{Siegel_2014} to our setting of an informed player $I$ and an uninformed player $U$. Assume the two players compete in an all-pay auction over an indivisible good, where the players $I$ and $U$ have the same type spaces as before. When type $t_i \in T_I$ is realized ($i\in\{1,\ldots,m\}$) with probability $p_i > 0$, player $I$'s valuation of the good is $\nu_I(t_i)$ and player $U$'s valuation is $\nu_U(t_i)$. Assume the types in $T_I$ are ordered such that 
\begin{equation}\label{eq:WM}
	 \nu_I(t_i) \geq \nu_I(t_k) 
\end{equation}
whenever $i < k$, $i,k \in \{1,\ldots,m\}$. In the two player all-pay auction with asymmetric information and valuations, player $I$ selects the distributions of bids $F_I : T_I \times \mbb{R}_+ \rightarrow [0,1]$ contingent on its type. We write $F_I(t_i)$ to refer to the distribution of bids in type $i$, $F_I(t_i,x)$ to refer to its value at $x\geq 0$, and $f_I$ to refer to its density function. As player $U$ has a single type, it selects a single distribution of bids, $F_U: \mbb{R}_+ \rightarrow [0,1]$. The best-response problems for each player are
\begin{equation}\label{eq:APA_BR}
	\begin{aligned}
		\max_{\substack{F_I(t_i) \\ i=1,\ldots,m}} \sum_{i=1}^m p_i \int_0^\infty \left[\nu_I(t_i)F_U - x \right] dF_I(t_i) \\
		\max_{F_U} \sum_{i=1}^m p_i \int_0^\infty \left[ \nu_U(t_i) F_I(t_i) - x \right] dF_U .
	\end{aligned}
\end{equation}
%and $U$'s best-response problem is
%\begin{equation}\label{eq:APA_BR_U}
%		\max_{F_U} \sum_{i=1}^m p_i \int_0^\infty \left[ V_U(t_i) F_I(t_i) - x \right] dF_U .\\
%\end{equation}

\subsection{Algorithm of \cite{Siegel_2014}}
An iterative procedure is formulated in \cite{Siegel_2014} to construct equilibrium mixed strategies satisfying the correspondences \eqref{eq:APA_BR}. We note that this procedure handles general player information structures in the two-player all-pay auction. In this paper, we restrict our attention to its application in our informed-uninformed setting. In particular, the constructed marginals are proven to be piecewise constant functions with finite support, with the possibility of having point masses placed at zero. We do not give further details of this algorithm due to space limitations and for ease of exposition. We refer the reader to \cite{Siegel_2014} for a general proof that the output distributions satisfy \eqref{eq:APA_BR}.

\subsection{Connection of General Lotto to all-pay auctions}
In the game $\text{GL}(X_I,X_U,V,\bm{p})$, player $U$'s Lagrangian at both ex-ante and interim levels may be written as
\begin{equation}\label{eq:lotto_BR_U}
	\begin{aligned}
		\sum_{j=1}^n \max_{F_U^j } \left[ \sum_{i=1}^m p_i \int_0^\infty \left( \frac{2 v_i^j}{\lambda_U} F_I^j(t_i) - x \right) dF_U^j \right] %+ \lambda_U X_U - \sum_{j=1}^n \sum_{i=1}^m p_i v_i^j
	\end{aligned}
\end{equation}
where $\lambda_U>0$ is the Lagrange multiplier on U's expected budget constraint \eqref{eq:lotto_constraint}, and we have removed constant additive and multiplicative terms in the expression that do not depend on the decision variables $\{F_U^j\}_{j=1}^n$. In a similar fashion, player $I$'s  Lagrangian maximization at the interim level may be written
\begin{equation}\label{eq:lotto_BR_I}
		\max_{\substack{F_I^j(t_i) \\ j=1,\ldots,n}} \lambda_I^i  \sum_{j=1}^n  \int_0^\infty\left( \frac{2v_i^j }{\lambda_I^i}  F_U^j - x \right) dF_I^j(t_i) 
\end{equation}
where the multiplier $\lambda_I^i >0$ corresponds to the budget constraint \eqref{eq:lotto_constraint} in type $t_i$. When two sets of battlefields contain the same valuations, we can deduce equivalence between the corresponding Lagrange multipliers.
\begin{lemma}\label{lambda_equivalence}
	Consider the game $\text{GL}(X_I,X_U,V,\bm{p})$. Suppose the rows $i$ and $k$ of $V$ have the same elements, each with identical multiplicities. Then the equilibrium ex-interim payoff to player I for type $t_i$ is equivalent to that of type $t_k$. Furthermore, $\lambda_I^{i} = \lambda_I^{k}$. 
\end{lemma}
\begin{proof}
	An equivalent formulation of \eqref{eq:lotto_BR_I} in type $i$ is 
	\begin{equation}\label{eq:interim_nolambda}
		\max_{F_I(t_i) \in \mcal{L}(X_I)} \pi_I(F_I,F_U | t_i)
	\end{equation}
	The corresponding problem for type $t_k$ is identical to the above, because the valuations in both rows are the same (possibly with some permutation of indices $j$), and the optimization for player U remains \eqref{eq:lotto_BR_U}. This shows the equivalence of  interim equilibrium payoffs.
	
	To show $\lambda_I^i = \lambda_I^k$, let $F_I^*$ solve \eqref{eq:interim_nolambda} for both types $t_i$ and $t_k$. Any allocation $x_j^I \in \mbb{R}_+$ to battlefield $j$ in the support of  $F_I^{j*}$ solves the one-dimensional problem $\max_{x}  \frac{2v_i^j }{\lambda_I^i}F_U^{j*} - x$ as well as $\max_{x}  \frac{2v_i^j }{\lambda_I^k}F_U^{j*} - x$. Therefore, the first-order necessary condition for optimality that holds is $\lambda_I^i = 2v_i^jf_U^{j*}(x_j^I) = \lambda_I^k$.
%\begin{equation}
%	\lambda_I^i = 2v_i^jF_U^{j*}(x_j^I) = \lambda_I^k .
%\end{equation}

\end{proof}
To solve for a BNE, we write the best-response correspondences at the ex-ante level.   Player U's ex-ante best-response problem is given by \eqref{eq:lotto_BR_U}, while Player I's ex-ante  optimization problem may be written
\begin{equation}\label{eq:lotto_BR_I_exante}
	 \sum_{j=1}^n \max_{\substack{F_I^j(t_i) \\ i=1,\ldots,m}} \left[ \sum_{i=1}^m \lambda_I^i  \int_0^\infty \left( \frac{2 v_i^j p_i}{\lambda_I^i} F_U^j - x \right) dF_I^j(t_i) \right]
\end{equation}
after removing additive constants. Under the conditions $\lambda_I^i = \lambda_I>0$  for all $i=1,\ldots,m$ (by Lemma \ref{lambda_equivalence}), each battlefield $j$ is an independent all-pay auction, whose problem is
\begin{align}
		&\max_{F_U^j } \left[ \sum_{i=1}^m p_i \int_0^\infty \left( \frac{2 v_i^j}{\lambda_U} F_I^j(t_i) - x \right) dF_U^j \right] \label{eq:BR_Uj} \\
		&\max_{\substack{F_I^j(t_i) \\ i=1,\ldots,m}} \left[  \sum_{i=1}^m p_i \int_0^\infty \left( \frac{2 v_i^jp_i }{\lambda_I} F_U^j - x \right) dF_I^j(t_i) \right] \label{eq:BR_Ij} 
\end{align}
which coincides with \eqref{eq:APA_BR} with auction valuations $\nu_U(t_i) = \frac{2 v_i^j}{\lambda_U}$ and $\nu_I(t_i) = \frac{2 v_i^j p_i }{\lambda_I}$. Here, we have multiplied each maximization problem of \label{eq:lotto_BR_I_exante} by $p_i > 0$, which does not change the optimal solutions.  In this setting, mixed-strategy equilibria of the Lotto game are equivalent to that of $n$ independent  two-player all-pay auctions with asymmetric information and valuations.

We may then apply the algorithm of \cite{Siegel_2014} to construct equilibrium distributions  $F_U^j$ and $\{ F_I^j(t_i)\}_{i=1}^m$ for each battlefield $j$. The constructed distributions are functions of the known parameters as well as the Lagrange multipliers. If there exists unique multipliers $\lambda_I^*,\lambda_U^*>0$ such that the  Lotto constraints \eqref{eq:lotto_constraint} is met for all types $\{t_I^i\}_{i=1}^m$ for $I$ and for type $t_U$ for $U$, then it is  clear the $n$ constructed strategy profiles constitute a BNE for  $\text{GL}(X_I,X_U,V,\bm{p})$. Indeed, we have applied the algorithm  to obtain a set of distributions, and verified there is a unique   $\lambda_I^*,\lambda_U^*>0$ such that \eqref{eq:lotto_constraint} is satisfied. The details of this calculation are outlined as follows. 
\begin{proof}[Proof of Theorem \ref{thm_lotto}]
	In $\text{GL}(X_I,X_U,V_{\alpha\beta},\frac{1}{3}\mbb{1}_3)$, we deduce from Lemma \ref{lambda_equivalence} that $\lambda_I^i = \lambda_I$ for all $i=1,2,3$. Hence we need only apply  the algorithm of \cite{Siegel_2014} to a single ``column" of the game to obtain all marginals of the BNE since the valuations are identical in each battlefield, i.e. the correspondences \eqref{eq:BR_Uj}, \eqref{eq:BR_Ij} are the same for all battlefields $j$, with a permutation of indices $i\in\{1,2,3\}$. 
	
Denote the constructed marginal distributions as $F_U$ for player $U$ and $\{F_I^{\text{d}}, F_I^{\alpha}, F_I^{\beta}\}$ for player $I$, where ``d" is for the diagonal battlefield value 1. We find that depending on whether $\frac{\lambda_I}{\lambda_U} \geq 1$,  $\frac{\lambda_I}{\lambda_U} \in (\frac{1}{2},1)$, or $\frac{\lambda_I}{\lambda_U} \in (\frac{1}{2},\frac{1}{3})$, the algorithm determines three distinct sets of marginal distributions for $I$ and $U$. In each case, there are unique $\lambda_I,\lambda_U>0$ such that the constraint \eqref{eq:lotto_constraint} is met. For brevity, we illustrate the calculation for one such case, as the other two follow similar methods. If $\frac{\lambda_I}{\lambda_U} \geq 1$,  the algorithm of \cite{Siegel_2014} gives
\begin{equation}
            F_U(x) = \frac{3\lambda_I}{2c}x,  \quad  x \in \left[0,\frac{2c}{3\lambda_I}\right]
\end{equation}
\begin{equation}
	\begin{array}{ll}
        		F_I^{\text{d}}(x) = \frac{3\lambda_U}{2c}x + 1 - \frac{\lambda_U}{\lambda_I}, & x \in \left[0,\frac{2c}{3\lambda_I}\right] \\
        		F_I^\alpha(x) = 1, & x \geq 0\\
        		F_I^\beta(x) = 1, & x \geq 0
        	  \end{array}
\end{equation}
The constraint \eqref{eq:lotto_constraint} requires that $\mbb{E}_{F_I^{\text{d}}} + \mbb{E}_{F_I^\alpha} + \mbb{E}_{F_I^\beta} = X_I$ and $3\mbb{E}_{F_U} = X_U$, from which we obtain the unique solutions $\lambda_I = \frac{1}{X_U(1+\alpha+\beta)}$ and $\lambda_U = 3\gamma\lambda_I$. This solution implies the budget ratio satisfies $\gamma \in (0,\frac{1}{3})$. Using these marginals, we calculate the equilibrium ex-ante payoff \eqref{eq:I_exante} as
\begin{equation}
	\frac{2}{1\!+\!\alpha\!+\!\beta}\left[ \int_0^\infty \!\!\! F_UdF_I^{\text{d}} \!+\! \alpha\int_0^\infty \!\!\!F_UdF_I^{\alpha} \!+\! \beta\int_0^\infty \!\!\! F_UdF_I^{\beta}\right] \!-\! 1.
	%\frac{2/3}{1\!+\!\alpha\!+\!\beta}\left[ \int_0^\infty \!\!\! F_UdF_I^{\text{d}} \!+\! \alpha\int_0^\infty \!\!\!F_UdF_I^{\alpha} \!+\! \beta\int_0^\infty \!\!\! F_UdF_I^{\beta}\right] \!-\! 1.
\end{equation}
From this, we obtain \eqref{eq:lotto_payoff1}. The other cases $\frac{\lambda_I}{\lambda_U} \in (\frac{1}{2},1)$ and $\frac{\lambda_I}{\lambda_U} \in (\frac{1}{2},\frac{1}{3})$ correspond to the budget ranges $\gamma \in [\frac{1}{3},\frac{2}{3})$ and $\gamma \in [\frac{2}{3},1)$, respectively.

For completeness, we present the marginals for the other two cases. When $\gamma \in [\frac{1}{3},\frac{2}{3})$, the unique solution of the multipliers are $\lambda_I \!=\! \frac{1}{X_U(1\!+\!\alpha\!+\!\beta)}\left[ \frac{1-\alpha}{9\gamma^2} + \alpha \right]$ and  $\lambda_U \!=\! 3\gamma \lambda_I$.
%	\begin{equation}
%		\begin{aligned}
%			\lambda_I &\!=\! \frac{1}{X_U(1\!+\!\alpha\!+\!\beta)}\left[ \frac{1-\alpha}{9\gamma^2} + \alpha \right] \\	 	
%			\lambda_U &\!=\! 3\gamma \lambda_I.
%		\end{aligned}
%	\end{equation}
Denote the intervals {\small$I_1=\left[0, \frac{2 c}{3}(\frac{\alpha}{\lambda_I} {-} \frac{\alpha}{\lambda_U})\right]$} and {\small$I_2=\left[\! \frac{2 c}{3}(\frac{\alpha}{\lambda_I} \! - \! \frac{\alpha}{\lambda_U}), \frac{2c}{3}(\frac{\alpha}{\lambda_I} {+} \frac{1-\alpha}{\lambda_U}) \!\right]$}. Then the equilibrium marginals are
 \begin{equation} 
        F_U(x)=\begin{cases}
        		\frac{3\lambda_I}{2\alpha c}x & x \in I_1 \\
    \frac{3\lambda_I}{2c}x \!+\! (1\!-\! \frac{\lambda_I}{\lambda_U})(1\!-\!\alpha) &x \in I_2
        	    \end{cases}
        	\end{equation}
	\begin{equation}
		\begin{array}{ll}
		F_I^{\text{d}}(x) = \frac{3\lambda_U}{2c}x - \alpha\!\left(\frac{\lambda_U}{\lambda_I} - 1\right) & x \in I_2 \\
		F_I^\alpha(x) = \frac{3\lambda_U}{2\alpha c} x + 2 - \frac{\lambda_U}{\lambda_I}  & x \in I_1 \\
		F_I^\beta(x) = 1, & x \geq 0
	  \end{array}
	\end{equation}
	When $\gamma \in [\frac{2}{3},1)$, the unique solution of the multipliers are $\lambda_I \!=\! \frac{c}{X_U}\left[ \beta + \frac{1}{9\gamma^2}(1 + 3\alpha - 4\beta) \right]$ and $\lambda_U \!=\! 3\gamma\lambda_I$.
%	\begin{equation}
%		\begin{aligned}
%			\lambda_I &\!=\! \frac{c}{X_U}\left[ \beta + \frac{1}{9\gamma^2}(1 + 3\alpha - 4\beta) \right]\\	 	
%			\lambda_U &\!=\! 3\gamma\lambda_I .
%		\end{aligned}
%	\end{equation}
Denote the intervals {\small$I_1=\left[0, \frac{2 c}{3}(\frac{\beta}{\lambda_I} {-} \frac{2\beta}{\lambda_U})\right]$}, {\small $I_2 = \left[ \frac{2 c}{3}\left(\frac{\beta}{\lambda_I} {-} \frac{2\beta}{\lambda_U}\right), \frac{2 c}{3}\left(\frac{\beta}{\lambda_I} {+} \frac{\alpha - 2\beta}{\lambda_U}\right) \right]$}, and {\small$I_3 = \left[ \frac{2 c}{3}\left(\frac{\beta}{\lambda_I} {+} \frac{\alpha - 2\beta}{\lambda_U}\right), \frac{2 c}{3}\left(\frac{\beta}{\lambda_I} {+} \frac{\alpha - 2\beta + 1}{\lambda_U}\right) \right]$}. Then the equilibrium marginals are calculated to be 
	\begin{equation}
		F_U(x) = \begin{cases}
    \frac{3\lambda_I}{2\beta c}x & x \in I_1 \\
    \frac{3\lambda_I}{2\alpha c}x + \left( 1 {-} 2\frac{\lambda_I}{\lambda_U}\right) \left(1 {-} \frac{\beta}{\alpha}\right) & x \in I_2 \\
    \frac{3\lambda_I}{2c}x + (1-\beta) + (2\beta - \alpha - 1) \frac{\lambda_I}{\lambda_U} &x \in I_3
    \end{cases}
	\end{equation}
	\begin{equation}
		\begin{array}{ll}
				F_I^{\text{d}}(x) = \frac{3\lambda_U}{2c}x - \beta\left(\frac{\lambda_U}{\lambda_I} - 2\right) - \alpha & x \in I_3 \\
		F_I^\alpha(x) = \frac{3\lambda_U}{2\alpha c} x - \frac{\beta}{\alpha}\left( \frac{\lambda_U}{\lambda_I} - 2\right) & x \in I_2 \\
		    F_I^\beta(x) = \frac{3\lambda_U}{2\beta c} x + 3 - \frac{\lambda_U}{\lambda_I}  & x \in I_1
		 \end{array}
	\end{equation}
\end{proof}

%%%%%%%%%%%%%%%%%%%%%%%%%%%%%%%%%%%%%%%%%%%%%%%%%%%%%%%%%%%%%%%%%%%%%%%%%%%
%%%%%%%%%%%%%%%%%%%%%%%%%%%%%  S E C T I O N 4 %%%%%%%%%%%%%%%%%%%%%%%%%%%%%%%%%%%%%%
%%%%%%%%%%%%%%%%%%%%%%%%%%%%%%%%%%%%%%%%%%%%%%%%%%%%%%%%%%%%%%%%%%%%%%%%%%%

%%%%%%%%%%%%%%%%%%%%%%%%%%%%%%%%%%%%%%%%%%%%%%%%%%%%%%%%%%%%%%%%%%%%%%%%%%%
%%%%%%%%%%%%%%%%%%%%%%%%%%%%%  S E C T I O N 5 %%%%%%%%%%%%%%%%%%%%%%%%%%%%%%%%%%%%%%
%%%%%%%%%%%%%%%%%%%%%%%%%%%%%%%%%%%%%%%%%%%%%%%%%%%%%%%%%%%%%%%%%%%%%%%%%%%
\section{Conclusion}\label{sec:conclusion}
In this paper, we extended the Colonel Blotto and General Lotto games to a setting where players have asymmetric information about the valuations of the battlefields. We focused on the case when one player is completely informed about the valuations but has fewer resources to allocate, and the other player is uninformed. Our analysis on the two battlefield case in the Colonel Blotto game shows  an informed player still cannot defeat its opponent in a mixed-strategy Nash equilibrium. We find a three battlefield scenario presents enough complexity such that the informed player in the General Lotto game can attain the advantage for certain parameters. 

A direction of future research involves generalizing the connection of the all-pay auctions with asymmetric information to the General Lotto game. This will allow us to investigate General Lotto games where the players hold arbitrary information structures.

%%%%%%%%%%%%%%%%%%%%%%%%%%%%%%%%%%%%%%%%%%%%%%%%%%%%%%%%%%%%%%%%%%%%%%%%%%%
%%%%%%%%%%%%%%%%%%%%%%%%%%%%%  A P P E N D I X %%%%%%%%%%%%%%%%%%%%%%%%%%%%%%%%%%%%%%
%%%%%%%%%%%%%%%%%%%%%%%%%%%%%%%%%%%%%%%%%%%%%%%%%%%%%%%%%%%%%%%%%%%%%%%%%%%
\appendix
\begin{proof}[Proof of Theorem \ref{blotto_phalf}]
Let $d := X_U - X_I$ and $q := \lfloor \frac{X_U}{X_U - X_I} \rfloor$ so that $X_U = qd + r$, where $0\leq r < d$. Denote a delta mass function centered at $y \in \mbb{R}$ by $\delta_y$. Define $c:=\vbar/\vubar$. We prove the Theorem by proposing a set of mixed strategy distributions $F_I^* = \{F_I^*(t_1), F_I^*(t_2)\}$, and $F_U^*$ each satisfying \eqref{eq:blotto_constraint}, and showing the strategy $F_z^*$ is a best-response to $F_{-z}^*$, $z=I,U$. For brevity, we prove the case when $q$ is odd, as the even case provides similar mixed strategies and follows similar arguments. Let $e \in (r,d)$, and consider the strategies
\begin{equation}\label{eq:FU_blotto}
	\begin{aligned}
		F_U^* = \frac{1}{s_A}&\Bigg( \sum_{k=1}^{\frac{q-1}{2}} c^{\frac{q+1}{2} - k} \delta_{e + (k-1)d} + \delta_{e + \frac{q-1}{2}d} \Bigg. \\
		&+ \Bigg. \sum_{k=\frac{q+1}{2}+1}^{q} c^{k - \frac{q+1}{2}} \delta_{e + (k-1)d} \Bigg)
	\end{aligned}
\end{equation}
\vspace{-5mm}
\begin{equation}\label{eq:FI_blotto}
		\begin{aligned}
			F_I^*(t_1) &=  \frac{1}{s_B} \Bigg( \frac{\underline{v}c^{\frac{q-1}{2}}}{\bar{v} + \underline{v}} \delta_{\frac{q-1}{2}d} + \sum_{k=\frac{q-1}{2} + 1}^{q-1} c^{q-1-k} \delta_{kd} \Bigg) \vspace{-10mm} \\
			F_I^*(t_2) &= \frac{1}{s_B} \Bigg( \sum_{k=0}^{\frac{q-1}{2}-1} c^k \delta_{kd} + \frac{\underline{v}c^{\frac{q-1}{2}}}{\bar{v} + \underline{v}} \delta_{\frac{q-1}{2}d} \Bigg) \\
		\end{aligned}
	\end{equation}
	where \vspace{-3mm}
	\begin{equation}
		s_A :=1 + 2\sum_{k=1}^{\frac{q-1}{2}} c^k, \ s_B := \frac{\underline{v}c^{\frac{q-1}{2}}}{\bar{v} + \underline{v}} + \sum_{k=0}^{\frac{q-1}{2}-1} c^k
	\end{equation}
	are normalizing factors. Before proceeding, we make a few remarks. These strategies are similar in nature to the strategies provided in Gross \& Wagner \cite{Gross_1950}. There, the authors show that a strategy composed of equally spaced delta functions with geometrically decreasing weights equalizes the payoff of the other player, and vice versa. In a similar fashion, the strategies \eqref{eq:FU_blotto} and \eqref{eq:FI_blotto} equalize the ex-ante payoffs in certain intervals of the players' allocation space. Any allocation in these intervals give a best-response to the other player's equilibrium strategy. 
	
	Let $x_U \in [0,X_U]$ be an allocation to battlefield 1, leaving $X_U-x_U$ to battlefield 2. 	The payoff $u_U(x_U,F_I^*(t_1),\omega_1)$ \eqref{eq:lotto_payoff} of any allocation $x_U$ against $F_I^*(t_1)$ in battlefield set 1 is 
	\begin{equation}{\small
		\begin{aligned}
			&\frac{1}{s_B}\left[\! \sum_{k=\frac{q-1}{2} \!+\! 1}^{q-1} c^{q\!-\!1\!-\!k}\left( \bar{v}\text{sgn}(x_U \!-\! kd) \!+\! \vubar\text{sgn}((k\!+\!1)d \!-\! x_U)\right) \right]  \\
			&+ \frac{1}{s_B}\left[\frac{\vubar c^{\frac{q-1}{2}}}{\bar{v} + \vubar}\! \left( \! \bar{v}\text{sgn}\left(x_U \!-\! \frac{q\!-\!1}{2}d\right) \!+\! \vubar \text{sgn}\left(\frac{q\!+\!1}{2}d \!-\! x_U\! \right) \right) \!\right]
		\end{aligned}}
	\end{equation}
	For notational purposes, let $H_U(\vbar,\vubar,x_U)$ denote the above quantity. Then against $F_I^*(t_2)$, $u_U(x_U,F_I^*(t_2),\omega_2) = H_U(\vubar,\vbar,x_U)$ in battlefield set 2.
%	\begin{equation}
%		\begin{aligned}
%			&\frac{1}{s_B}\left[ \!\sum_{k=0}^{\frac{m-1}{2} \!-\!1} c^k \left(\vubar\text{sgn}(x_U \!-\! kd ) + \bar{v} \text{sgn}((k+1)d - x_U) \right)\right] \\
%			&\!\!+\!\frac{1}{s_B}\!\!\left[ \frac{\underline{v} c^{\frac{q-1}{2}}}{\bar{v} + \underline{v}} \left( \underline{v}\text{sgn}\left(\!x_U \!-\! \frac{q\!-\!1}{2}d\right) \!\!+\! \bar{v}\text{sgn}\left(\!\frac{q\!+\!1}{2}d \!-\! x_U\right) \!\right)\!\right] 
%		\end{aligned}
%	\end{equation}
	Given $p=\frac{1}{2}$, the ex-ante utility is $\pi_U(x_U,F_I^*) = \frac{1}{2}u_U(x_U,F_I^*(t_1),\omega_1) + \frac{1}{2}u_U(x_U,F_I^*(t_2),\omega_2)$. After some algebra, we arrive at 
	 \begin{equation}
	 	\pi_U(x_U,F_I^*) = 
		\begin{cases} 
			\frac{\vbar}{s_B} &\text{if } x_U \in (0,qd \ ] \\
			0 &\text{if } x_U \in (qd,X_U]
		\end{cases}
	 \end{equation}
	 Any mixed strategy $F_U$ with support on the interval $(0,qd)$ is a best-response to $F_I^*$, and hence $F_U^* \in \text{BR}_U(F_I^*)$. Now, any pure allocation $x_I^1 \in [0,X_I]$ of player $I$ against $F_U^*$ in type $t_1$ gives the ex-interim utility $\pi_I(x_I^1,F_U^* | t_1)$ \eqref{eq:interim_I}, as
	 \begin{equation}{\small
	 	\begin{aligned}
			&\!\!\frac{-1}{s_A}\!\left[ \sum_{k=1}^{\frac{q\!-\!1}{2}} c^{\frac{q+1}{2} \!-\! k}\!\! \left( \!\bar{v}\text{sgn}(e \!+\! (k\!-\!1)d \!-\! x_I^1) \!-\! \underline{v}\text{sgn}(e \!+\! (k\!-\!2)d \!-\! x_I^1) \right)\! \right. \\
			&+ \left( \!\bar{v}\text{sgn}\left(e \!+\! \frac{q-1}{2}d \!-\! x_I^1\right) - \underline{v}\text{sgn}\left(e \!+\! \left(\frac{q\!-\!1}{2} \!-\! 1\right)d \!-\! x_I^1\right) \!\right) \\
			&\!\!\! + \!\!\!\left. \!\!\sum_{k=\frac{q+1}{2}+1}^{q}\!\!\!\!\! c^{k - \frac{q+1}{2}} \left( \bar{v}\text{sgn}(e \!+\! (k\!-\!1)d \!-\! x_I^1) \!-\! \underline{v}\text{sgn}(e \!+\! (k\!-\!2)d \!-\! x_I^1) \right)\! \right]\\
		\end{aligned}}
	\end{equation}
	For notational purposes, let $H_I(\vbar,\vubar,x_I^1)$ denote the above quantity. 
%	We can rearrange terms to express $H_I(\vbar,\vubar,x_I^1)$ as
%	\begin{equation}{\small
%		\begin{aligned}
%			&\frac{-\underline{v}}{s_A}\left[ \!\!\sum_{k=2}^{\frac{q-1}{2}-1}\!\!\! c^{\frac{q+1}{2} - k} \left( \!\text{sgn}(e \!+\! kd \!- \!x_I^1) \!-\! \text{sgn}(e\!+\!(k\!-\!2)d \!-\! x_I^1) \right)  \right. \\
%			& \quad\quad\!+\!  c^{\frac{q-1}{2}}(1 + \text{sgn}(e+d-x_I^1) ) +  c^{\frac{q+1}{2}}(1\!+\!\text{sgn}(e-x_I^1)) \\ 
%			& \!\!\left.\! -c\text{sgn}\!\!\left(\!\!e \!+\!\! \left( \!\! \frac{q-1}{2} \!-\! 2 \right)\!d \!-\! x_I^1 \! \right)  \!-\! \text{sgn}\left(\!e \!+\!\! \left(\!\!\frac{q-1}{2} \!-\! 1\right)\!d \!-\! x_I^1\! \right) \!  \right] \\
%		\end{aligned}}
%	\end{equation}
	One can list all possible values this takes as a function of $x_I^1$ (not shown due to space constraint). This is increasing in $x_I^1$, and attains the maximum value when {\small$x_I^1 \in (e+\left(\frac{q-1}{2} - 1 \right)d, X_I)$}, giving  $\max_{x_I^1} \pi_I(x_I^1,F_U^* | t_1) =  -\frac{\underline{v}}{s_A}(1+c)$.
%	\begin{equation}
%		\max_{x_I^1} \pi_I(x_I^1,F_U^* | t_1) =  -\frac{\underline{v}}{s_A}(1+c)
%	\end{equation}

	Any pure allocation $x_I^2 \in [0,X_I]$ of player $I$ against $F_U^*$ in type $t_2$ gives ex-interim utility $\pi_I(x_I^2,F_U^* | t_2) = H_I(\vubar,\vbar,x_I^2)$. Through a similar analysis, we find the maximal value $-\frac{\underline{v}}{s_A}(1+c)$ is attained for {\small$x_I^2 \in (0,e + \frac{q-1}{2}d)$}. Hence, any mixed strategy $F_I(t_1)$ with support on {\small$(e+\left(\frac{q-1}{2} - 1 \right)d, X_I)$} and $F_I(t_2)$ with support on $(0,e + \frac{q-1}{2}d)$ is a best-response to $F_U^*$. Consequently, $F_I^* \in \text{BR}_I(F_U^*)$. The quantities $\frac{\vubar}{s_B}$ and $\frac{\underline{v}}{s_A}(1+c)$ coincide, giving the value of the game.
	
	\end{proof}
%\section*{Acknowledgements}


%%%%%%%%%%%%%%%%%%%%%%%%%%%%%%%%%%%%%%%%%%%%%%%%%%%%%%%%%%%%%%%%%%%%%%%%%%%%%%%%

%%%%%%%%%%%%%%%%%%%%%%%%%%%%%%%%%%%%%%%%%%%%%%%%%%%%%%%%%%%%%%%%%%%%%%%%%%%%%%%%

%%%%%%%%%%%%%%%%%%%%%%%%%%%%%%%%%%%%%%%%%%%%%%%%%%%%%%%%%%%%%%%%%%%%%%%%%%%%%%%%

%%%%%%%%%%%%%%%%%%%%%%%%%%%%%%%%%%%%%%%%%%%%%%%%%%%%%%%%%%%%%%%%%%%%%%%%%%%%%%%%

\bibliographystyle{IEEEtran}
\bibliography{sources}

\begin{thebibliography}{10}
\providecommand{\url}[1]{#1}
\csname url@samestyle\endcsname
\providecommand{\newblock}{\relax}
\providecommand{\bibinfo}[2]{#2}
\providecommand{\BIBentrySTDinterwordspacing}{\spaceskip=0pt\relax}
\providecommand{\BIBentryALTinterwordstretchfactor}{4}
\providecommand{\BIBentryALTinterwordspacing}{\spaceskip=\fontdimen2\font plus
\BIBentryALTinterwordstretchfactor\fontdimen3\font minus
  \fontdimen4\font\relax}
\providecommand{\BIBforeignlanguage}[2]{{%
\expandafter\ifx\csname l@#1\endcsname\relax
\typeout{** WARNING: IEEEtran.bst: No hyphenation pattern has been}%
\typeout{** loaded for the language `#1'. Using the pattern for}%
\typeout{** the default language instead.}%
\else
\language=\csname l@#1\endcsname
\fi
#2}}
\providecommand{\BIBdecl}{\relax}
\BIBdecl

\bibitem{Behnezhad_2018}
S.~Behnezhad, A.~Blum, M.~Derakhshan, M.~HajiAghayi, M.~Mahdian, C.~H.
  Papadimitriou, R.~L. Rivest, S.~Seddighin, and P.~B. Stark, ``From
  battlefields to elections: Winning strategies of blotto and auditing games,''
  in \emph{Proc. of the Twenty-Ninth Annual ACM-SIAM Symposium on Discrete
  Algorithms}.\hskip 1em plus 0.5em minus 0.4em\relax SIAM, 2018, pp.
  2291--2310.

\bibitem{Thomas_2018}
C.~Thomas, ``N-dimensional blotto game with heterogeneous battlefield values,''
  \emph{Economic Theory}, vol.~65, no.~3, pp. 509--544, 2018.

\bibitem{Ferdowsi_2017}
A.~Ferdowsi, W.~Saad, B.~Maham, and N.~B. Mandayam, ``A colonel blotto game for
  interdependence-aware cyber-physical systems security in smart cities,'' in
  \emph{Proceedings of the 2nd International Workshop on Science of Smart City
  Operations and Platforms Engineering}.\hskip 1em plus 0.5em minus 0.4em\relax
  ACM, 2017, pp. 7--12.

\bibitem{Fazeli_2017}
A.~Fazeli, A.~Ajorlou, and A.~Jadbabaie, ``Competitive diffusion in social
  networks: Quality or seeding?'' \emph{IEEE Transactions on Control of Network
  Systems}, vol.~4, no.~3, pp. 665--675, 2017.

\bibitem{Li_2017}
L.~Li and J.~S. Shamma, ``Efficient strategy computation in zero-sum asymmetric
  repeated games,'' \emph{arXiv preprint arXiv:1703.01952}, 2017.

\bibitem{Kartik_2019}
D.~Kartik and A.~Nayyar, ``Zero-sum stochastic games with asymmetric
  information,'' \emph{arXiv preprint arXiv:1909.01445}, 2019.

\bibitem{Borel}
E.~Borel, ``La th\'eorie du jeu les \'equations int\'egrales \`a noyau
  sym\'etrique,'' \emph{Comptes Rendus de l'Acad\'emie}, vol. 173, 1921.

\bibitem{Gross_1950}
O.~Gross and R.~Wagner, ``A continuous colonel blotto game,'' RAND PROJECT AIR
  FORCE SANTA MONICA CA, Tech. Rep., 1950.

\bibitem{Roberson_2006}
B.~Roberson, ``The colonel blotto game,'' \emph{Economic Theory}, vol.~29,
  no.~1, pp. 1--24, 2006.

\bibitem{Sklar_1973}
A.~Sklar, ``Random variables, joint distribution functions, and copulas,''
  \emph{Kybernetika}, vol.~9, no.~6, pp. 449--460, 1973.

\bibitem{Hart_2008}
S.~Hart, ``Discrete colonel blotto and general lotto games,''
  \emph{International Journal of Game Theory}, vol.~36, no. 3-4, pp. 441--460,
  2008.

\bibitem{Macdonell_2015}
S.~T. Macdonell and N.~Mastronardi, ``Waging simple wars: a complete
  characterization of two-battlefield blotto equilibria,'' \emph{Economic
  Theory}, vol.~58, no.~1, pp. 183--216, Jan 2015.

\bibitem{Schwartz_2014}
G.~{Schwartz}, P.~{Loiseau}, and S.~S. {Sastry}, ``The heterogeneous colonel
  blotto game,'' in \emph{2014 7th International Conf. on NETwork Games,
  COntrol and OPtimization (NetGCoop)}, Oct 2014, pp. 232--238.

\bibitem{Kovenock_2015}
D.~Kovenock and B.~Roberson, ``Generalizations of the general lotto and colonel
  blotto games,'' 2015.

\bibitem{Ferdowsi_2018}
A.~Ferdowsi, A.~Sanjab, W.~Saad, and T.~Basar, ``Generalized colonel blotto
  game,'' in \emph{2018 Annual American Control Conference (ACC)}.\hskip 1em
  plus 0.5em minus 0.4em\relax IEEE, 2018, pp. 5744--5749.

\bibitem{Shahrivar_2014}
E.~M. Shahrivar and S.~Sundaram, ``Multi-layer network formation via a colonel
  blotto game,'' in \emph{2014 IEEE Global Conference on Signal and Information
  Processing (GlobalSIP)}.\hskip 1em plus 0.5em minus 0.4em\relax IEEE, 2014,
  pp. 838--841.

\bibitem{Guan_2019}
S.~Guan, J.~Wang, H.~Yao, C.~Jiang, Z.~Han, and Y.~Ren, ``Colonel blotto games
  in network systems: Models, strategies, and applications,'' \emph{IEEE
  Transactions on Network Science and Engineering}, 2019.

\bibitem{Adamo_2009}
T.~Adamo and A.~Matros, ``A blotto game with incomplete information,''
  \emph{Economics Letters}, vol. 105, no.~1, pp. 100--102, 2009.

\bibitem{Kovenock_2011}
D.~Kovenock and B.~Roberson, ``A blotto game with multi-dimensional incomplete
  information,'' \emph{Economics Letters}, vol. 113, no.~3, pp. 273--275, 2011.

\bibitem{Siegel_2014}
R.~Siegel, ``Asymmetric all-pay auctions with interdependent valuations,''
  \emph{Journal of Economic Theory}, vol. 153, pp. 684 -- 702, 2014.

\bibitem{Vega-Redondo_2003}
F.~Vega-Redondo, \emph{Economics and the Theory of Games}.\hskip 1em plus 0.5em
  minus 0.4em\relax Cambridge University Press, 2003.

\end{thebibliography}

\end{document}